\documentclass[journal]{IEEEtran}
\usepackage{mathrsfs}
\usepackage{amsmath}
\usepackage{amssymb}
\usepackage{changes}
\usepackage{mathtools} 
\usepackage{amsthm}
\usepackage{verbatim}
\usepackage{amsfonts}
\usepackage{cite}
\usepackage{algorithm}
\usepackage{multirow}
\usepackage{subfigure}
\usepackage{algpseudocode}
\usepackage{booktabs}
\usepackage{amstext}
\usepackage{bm}
\usepackage{threeparttable}
\usepackage{nomencl}

\usepackage{url}

\allowdisplaybreaks[4]
\usepackage{float}
\usepackage{graphicx}

\DeclareGraphicsExtensions{.pdf,.png,.jpg,.eps,.ps}

\newcommand{\RNum}[1]{\uppercase\expandafter{\romannumeral #1\relax}}

\def\baa{\begin{align}}
\def\eaa{\end{align}}

\newcommand{\bsq}{\begin{subequations}}
	\newcommand{\esq}{\end{subequations}}

\newcommand{\beq}{\begin{equation}}
\newcommand{\eeq}{\end{equation}}
\newcommand{\bq}{\begin{eqnarray}}
\newcommand{\eq}{\end{eqnarray}}
\newcommand{\bqn}{\begin{eqnarray*}}
	\newcommand{\eqn}{\end{eqnarray*}}
\newcommand{\bee}{\begin{enumerate}}
	\newcommand{\eee}{\end{enumerate}}
\newcommand{\bi}{\begin{itemize}}
	\newcommand{\ei}{\end{itemize}}

\usepackage{comment}
\usepackage{color}
\newboolean{showcomments}
\setboolean{showcomments}{true}
\newcommand{\wang}[1]{\ifthenelse{\boolean{showcomments}}
	{ \textcolor[rgb]{1,0,1}{(ZW:  #1)}}{}}
\newcommand{\fliu}[1]{\ifthenelse{\boolean{showcomments}}
	{ \textcolor{red}{(FL:  #1)}}{}}
\newcommand{\zhang}[1]{\ifthenelse{\boolean{showcomments}}
	{ \textcolor{blue}{(YFZ:  #1)}}{}}

\theoremstyle{definition}
\newtheorem{theorem}{Theorem}
\newtheorem{lemma}[theorem]{Lemma}

\newtheorem{condition}{Condition}

\theoremstyle{definition}
\newtheorem{definition}{Definition}
\newtheorem{remark}{Remark}

\newtheorem{assumption}{\textit{Assumption}}

\begin{document}

\title {Asynchrony-Resilient and Privacy-Preserving Charging Protocol for Plug-in Electric Vehicles}

\author{Yunfan ~Zhang,
        ~Feng ~Liu,~\IEEEmembership{Senior Member,~IEEE,}
        ~Zhaojian ~Wang,        
        ~Jianhui ~Wang, ~\IEEEmembership{Senior Member,~IEEE,}\\
        ~Yifan ~Su,
        ~Yue ~Chen,
        ~Cheng ~Wang,
        ~and~Qiuwei~Wu
}
\markboth{Journal of \LaTeX\ Class Files,~Vol.~xx, No.~xx, August~xxxx}%
{Shell \MakeLowercase{\textit{et al.}}: Bare Demo of IEEEtran.cls for IEEE Journals}

\maketitle

\begin{abstract}
	
The proliferation of plug-in electric vehicles (PEVs) advocates a distributed paradigm for the coordination of PEV charging.  Distinct from existing primal-dual decomposition or consensus methods, this paper proposes a cutting-plane based distributed algorithm, which enables an asynchronous coordination while well preserving individual's private information. To this end, an equivalent surrogate model  is first constructed  by exploiting the duality of the original optimization problem, which masks the private information of individual users by a transformation. Then, a cutting-plane  based algorithm is derived to solve the surrogate problem in a  distributed manner with intrinsic superiority to cope with various asynchrony. Critical implementation issues, such as the distributed initialization, cutting-plane generation and localized stopping criteria,  are discussed in detail. Numerical tests on IEEE 37- and 123-node feeders with real data show that the proposed method is resilient to a variety of asynchrony and admits the plug-and-play operation mode. It is expected the proposed methodology provides an alternative path toward a more practical protocol for PEV charging.  
\end{abstract}


\IEEEpeerreviewmaketitle

\section{Introduction}
\subsection{Background and Motivation}
The past years  witnessed the proliferation of plug-in electric vehicles (PEVs). However, their rapid growth inevitably creates new challenges to power system operation. Particularly, as traditional distribution systems were not  designed to support simultaneous charging of many PEVs\cite{Gan2013Optimal},  transformer capacity expansion or even reconstruction of the distribution system are needed to meet the growing demand for PEV charging \cite{ref3}. However, such costly countermeasures could be alleviated or even avoided if  the charging behavior of PEVs are well managed\cite{Gan2013Optimal},  in either a centralized or distributed manner.   

Traditional centralized management of PEV charging needs to collect all PEVs' information, such as positions, available charging time and state of charge (SOC), etc. Hence it may raise severe privacy concerns in individual PEV owners \cite{Engel2013Privacy}. Moreover, the  decision center may suffer from a heavy communication burden and high computational complexity. In this regard, distributed management was developed, where the charging patterns of PEVs are decided locally following a certain coordination scheme. It is expected to better protect the privacy of PEV owners and enable a faster response to environmental changes\cite{Omran2017A}, which is crucially important when numerous  PEVs  disperse across the distribution network. However, there always exist various kinds of asynchrony in practice due to non-ideal communication such as time delay and packet drop. These regards motivate us to address the  protocol of PEV distributed charging in this paper, considering asynchrony resilience and privacy-preserving.    

\subsection{Related Works}
Generally, prior works on distributed PEV charging management can be cast into two branches: non-cooperative strategies and cooperative ones
, which are briefly reviewed as follows.

\subsubsection{ Non-cooperative Strategies}
Non-cooperative charging strategies oftentimes are  partially distributed, where a coordinator is needed to broadcast coordination signals (usually electricity price) and then each PEV reacts to its received signals. A typical distributed non-cooperative strategy relies on  a one-way communication broadcast\cite{Turitsyn2010Robust}, but this open-loop approach appears to be less effective due to the absence of feedback adjustments\cite{Lopes2009Identifying, Vaya2012Centralized}. To address this problem, iterative strategies emerge to search for the optimal or quasi-optimal charging profiles, where  bi-directional communication between the coordinator and  individual PEVs is required and each PEV solves a restricted sub-problem in every round. Under this framework, several distributed charging algorithms are derived, based on Lagrangian dual decomposition\cite{Jiang2013Decentralized}, projected gradient\cite{Liu2017Decentralized,Ardakanian2013Distributed}, Alternating Direction Method of Multipliers (ADMM)\cite{Liang2016Scalable, Vaya2015Decentralized, Kraning2014Dynamic, Rivera2013Alternating} and theory of non-cooperative games\cite{Ma2012Decentralized,Parise2014Mean}, to name a few.

\subsubsection{Cooperative Strategies}
Cooperative strategies are usually investigated under a fully distributed framework, where PEVs collaborate (usually with their immediate neighbors) to achieve a certain optimal target\cite{Logenthiran2012Multi}. Such peer-to-peer (P2P) based schemes\cite{MohammadiA} serve as more flexible, scalable and robust alternatives since individual PEVs can autonomously achieve coordination with the absence of a coordinator. In this regard, distributed solution algorithms based on Karush-Kuhn-Tucker (KKT) conditions and consensus techniques are designed, see, for example, \cite{Rahbari2014Cooperative,Xu2015Optimal,MohammadiA}. 

The potential of P2P based cooperative schemes, however, has not been well addressed yet when it comes to the following two critical implementation issues: \romannumeral1) User-state-information (USI) privacy. PEV users are reluctant to disclose their USI (such as SOC, positions, and demand profiles, etc.) neither to a center nor to other users. 
In regard to the information exchange which is necessary for coordination in a P2P network, privacy issue also remains as a concern. \romannumeral1) Imperfect communication. Considering time delays, packet drops, topology changes and non-identical computation capabilities of individual users, the participants have to wait for the slowest one to finish before executing their local actions in the next iteration. Though prior works  \cite{Low1999,Wang2018D,HaleAsynchronous} have achieved many successes on distributed algorithms  under asynchronous communication, their convergence results and solution quality rely on restrictive assumptions. Moreover, it is not trivial  to select appropriate parameters and adjust the step size in the updating process. 

\subsubsection{Distributed Decision-making based on Cutting Planes}
The cutting-plane theory \cite{Eaves1971G}  has recently drawn increasing attention in the community of distributed decision-making. However, there are few works considering its application to  power systems, except \cite{Liu2016Fully} and \cite{ref17} that discuss dynamic economic dispatch and microgrid control, respectively. Different from traditional consensus algorithms, a cutting-plane based consensus algorithm can achieve an agreement on a common query point  through  iterative constraint exchanges. This salient feature can better decompose the computation  into individual agents with the minimal requirements of information synchrony, which inspires a suitable framework for the distributed PEV charging coordination against   asynchrony. 

\subsection{Contribution}
In this paper, a novel distributed optimal scheduling method as well as its implementation are proposed for individual PEV charging coordination under both local and global constraints, considering communication asynchrony and privacy preserving. The main contributions of this paper are threefold.

1) \emph{Asynchrony-Resilience.} Distinct from the celebrated  primal-dual decomposition methods \cite{Low1999,Wang2018D,HaleAsynchronous} and ADMM  algorithms \cite{Liang2016Scalable, Vaya2015Decentralized, Kraning2014Dynamic, Rivera2013Alternating, distributedADMM2} where   primal/dual variables are exchanged, the proposed  method turns to  exchange cutting planes among neighboring individuals. The resulting algorithm intrinsically   admits  asynchronous implementations, favoring a strong resilience to various asynchrony in practice such as time delays, packet drops and communication  topology changes. 
 
2) \emph{Privacy-Preserving.} Regarding privacy concerns, this paper  uses a surrogate model to mask the private information of
individual users during  the iterative coordination process. 
Different from existing works where the charging profiles or multipliers of the corresponding optimization problem are collected/exchanged  \cite{Jiang2013Decentralized,Liu2017Decentralized,Ardakanian2013Distributed,Liang2016Scalable, Vaya2015Decentralized, Kraning2014Dynamic, Rivera2013Alternating}, the proposed method only delivers aggregated information of the surrogate model, which can well protect the user-state  privacy of individual PEV owners.

3) \emph{Convergence-guarantee.}
 The existing cutting-plane based works heuristically take  any locally converged consensus result  as the optimal solution \cite{Liu2016Fully,ref17}, which lacks of a theoretic convergence guarantee. 
 This paper first unfolds the solution quality of the algorithm and derives  a completely localized stopping criteria. It is  proved that the solution series obtained by the proposed method must converge to the global optimum. Moreover,  the localized stopping criteria can achieve an arbitrarily small error.


\subsection{Organization}
The rest of the paper is organized as follows. The problem description with necessary preliminaries and notations is stated in Section II and Section III derives an equivalent surrogate model. Section IV presents the distributed solution algorithm based on cutting-plane consensus. Case studies are introduced in Section V. Finally, Section VI concludes the paper.

\section{Notation and Problem Formulation }
\subsection{Preliminaries and Notations}

In this paper, $\mathbb{R}^n$ ($\mathbb{R}^n_{+}$) depicts the $n$-dimensional (non-negative) Euclidean space. Use $ \mathbb{Z}^+ $ to denote the set of positive integers. For a column vector $z\in \mathbb{R}^n$ (matrix $A\in \mathbb{R}^{m\times n}$), $z^{\mathsf{T}}$($A^{\mathsf{T}}$) denotes its transpose. $\Vert\cdot\Vert$ denotes the Euclidean norm. Given a collection of $y_i$ for $i$ in a certain set $\mathcal{N}$, define $\text{col}(y_{j}):=(y_1, y_2, \cdots, y_n)^\mathsf{T}$ and denote its vector form by ${y}:=\text{col}(y_i)$. We use $\bm{1}$ (resp. $\bm{0}$) to denote vector of ones (resp. zeros). Notations for cutting-plane and cutting-plane set are given in Definition \ref{def1} and Definition \ref{def2} respectively.
\begin{definition}[Cutting-Plane]\label{def1}
	Given a convex set $\mathcal{S}\subset\mathbb{R}^n$ and a query point $z_q\notin\mathcal{S}$, a half-space
	\begin{align} 
	h_{z_q}:=\left\{z|a^{\mathsf{T}}_{z_q}z\le b_{z_q}\right\},a_{z_q}\in\mathbb{R}^n,b_{z_q}\in\mathbb{R}^1
	\end{align}
	is referred to as the cutting-plane of $\mathcal{S}$ and $z_q$ if it satisfies the following properties: (i) $a_{z_q}\neq \bm{0}$, (ii) $a^{\mathsf{T}}_{z_q}z_q> b_{z_q}$ and (iii) $a^{\mathsf{T}}_{z_q}z\le b_{z_q},\forall z\in\mathcal{S}$.
\end{definition}
\begin{definition}[Cutting-Plane Set]\label{def2}
	Cutting-plane set $H$ is the collection of $m$ single cutting-planes. Specifically,
	\begin{align}
	H:=\cup_{k=1}^m h_k
	\end{align}
	where $h_k:=\left\{z|a_k^{\mathsf{T}}z\le b_k \right\}$. The induced polyhedron of $H$ is denoted by $\mathcal{H}:=\left\{z|A^{\mathsf{T}}_{H}z\le b_{H}\right\}$, with $A_{H}:=[a_1,\ldots,a_m]$ and $b_H:=[b_1,\ldots,b_m]^{\mathsf{T}}$. 
\end{definition}
Note that in this paper the union symbol (same for intersection symbol) plays an opposite role on cutting-plane set and the induced polyhedron. For example, given two cutting-plane sets $H_1$ and $H_2$, we have $H=H_1\cup H_2\Rightarrow\mathcal{H}=\mathcal{H}_1\cap\mathcal{H}_2$, where $\mathcal{H}_1$, $\mathcal{H}_2$ and $\mathcal{H}$ are the polyhedrons directly induced by $ H_1 $,$ H_2 $ and $H$, respectively. For convenience of phrasing, let $h^{\emptyset}:=\left\{z|\bm{0}^{\mathsf{T}}z\le 0\right\}$ denote an empty cutting-plane.

\subsection{Charging of PEVs}

\begin{figure}[b]
	\centering
	\includegraphics[scale=0.3]{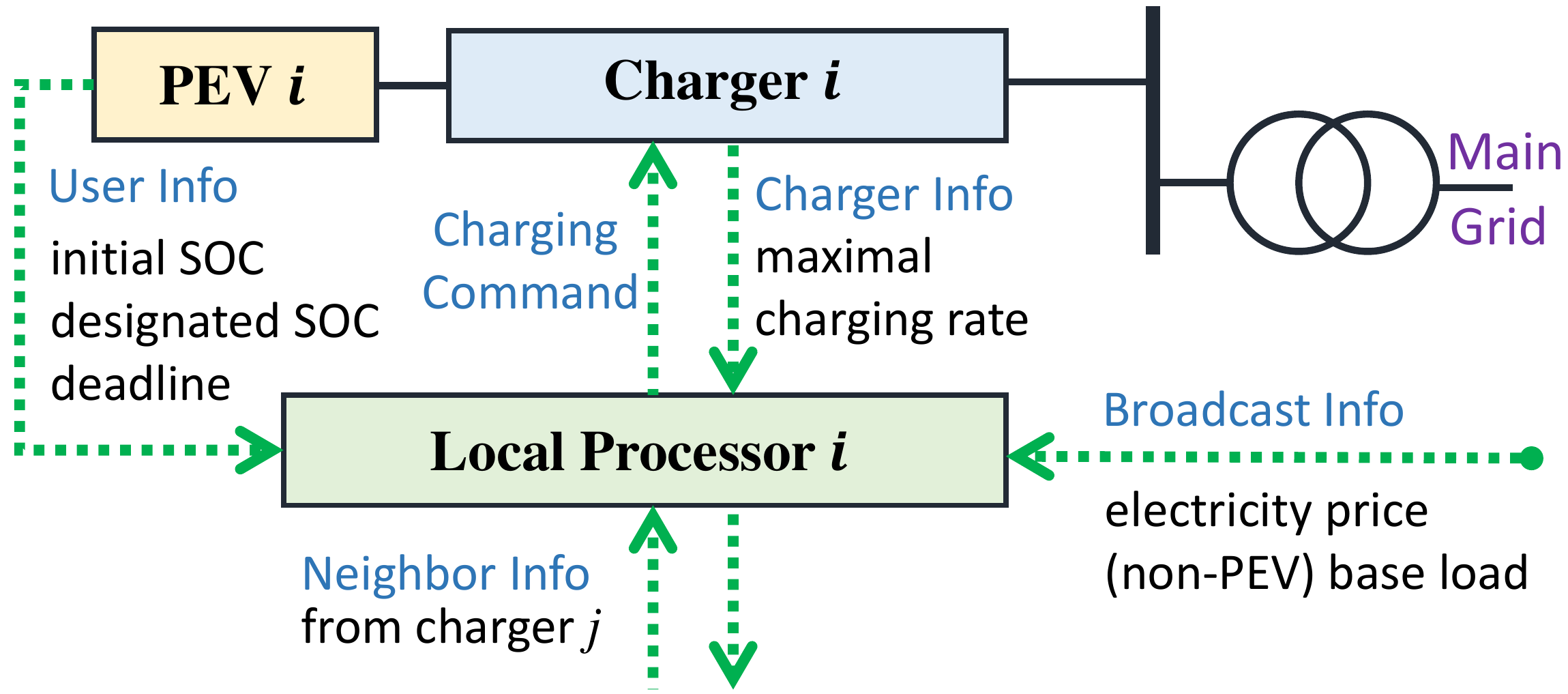}
	\caption{The schematic of the proposed local processor.}
	\label{facilities}
\end{figure}

To coordinate PEV charging in a fully distributed manner, we adopt the configuration shown in Fig.\ref{facilities}. Each PEV charger is equipped with a local processor with a certain capability of communication and computation. The processor collects data from the user (such as the designated SOC and charging deadline) and data from the PEV (such as the initial SOC and capacity of batteries). It also receives some broadcast information such as the electricity price. The local processor enables a bi-directional information exchange with its neighbors. Once the optimal charging strategy is derived, the control command is generated and sent to the PEV charger. 

\subsection{Communication network }
We consider a generalized asynchronous communication protocol presented in \cite{Wang2018D}, where each processor has its own concept of time defined by a \emph{local clock} $k_i\in\mathbb{Z}^{+}$. $k_i$ triggers when processor $i$ awakes, independently of other processors, to conduct local computations and update information to its neighbors. When $i$ is idle, it listens for messages from neighbors and stores them to its receiving cache. Then let $k\in\mathbb{Z}^{+}$ denote a virtual \emph{global clock} that does not exist in reality and is used only for analysis. The relationship between local clocks and the global clock is depicted in fig. \ref{clock}. 
\begin{figure}[!ht]
	\centering
	\includegraphics[scale=0.35]{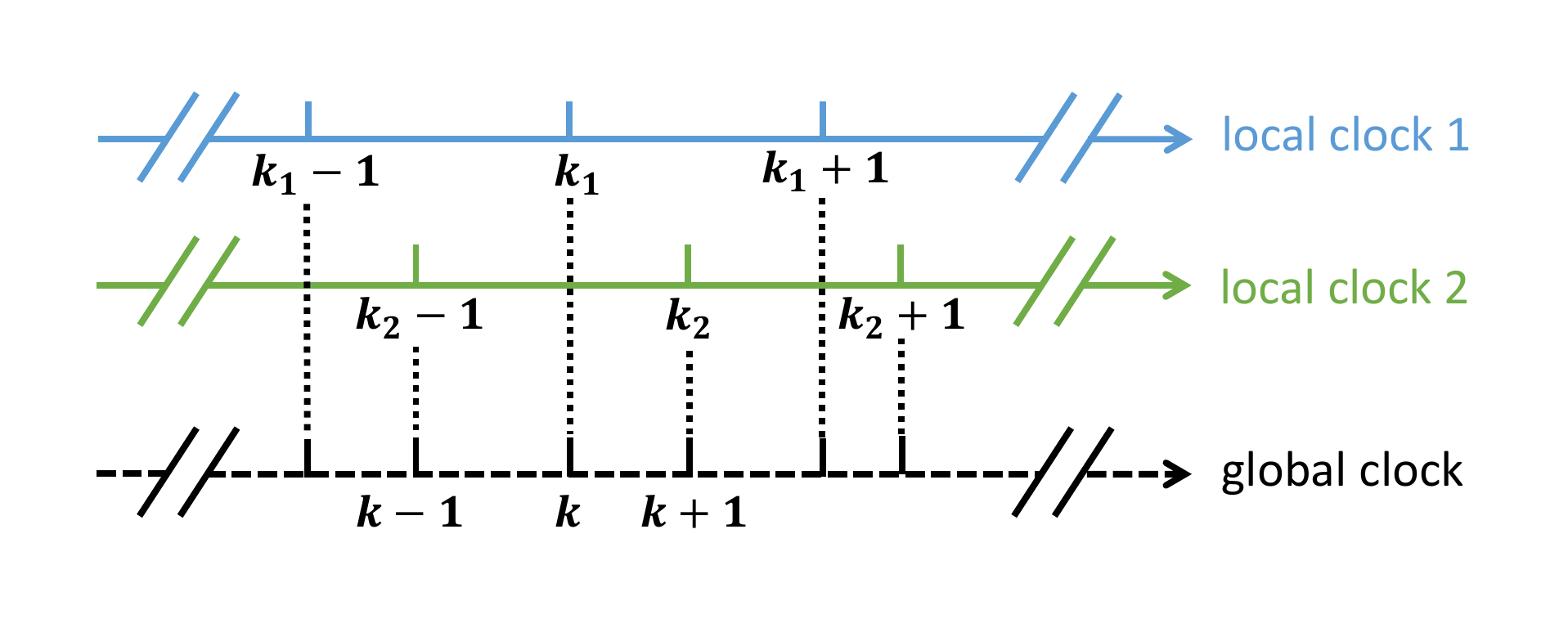}
	\caption{An instance of two local clocks and the virtual global clock. 
	}
	\label{clock}
\end{figure}

To consider the change of communication topology, we model the communicating network at the virtual global clock $k$ as a $k-$dependent directed graph $\mathcal{G}_k=(\mathcal{N},\mathcal{E}_k)$, where $\mathcal{N}:=\left\{1,2,\ldots,n\right\}$ is the set of local charger processors and $\mathcal{E}_k\subset \mathcal{N}\times \mathcal{N}$ is the edge set. If processor
$i$ transmits messages to processor $j$ at $k$, there is an edge from node $i$ to $j$ at $k$, denoted by $(i,j)\in \mathcal{E}_k$. For an edge set $\mathcal{E}_k$, we denote by $\mathcal{N}^{-}_{i,\mathcal{E}_k}:=\{j|(i,j)\in \mathcal{E}_k\}$ the out-neighbors of processor $i$ and by $\mathcal{N}^{+}_{i,\mathcal{E}_k}:=\{j|(j,i)\in \mathcal{E}_k\}$ the in-neighbors.
Also, for the $k-$dependent graph $\mathcal{G}_k$, we denote by $d_{\mathcal{G}_k}$ its diameter and $d_{\mathcal{G}_k}$ equals to the diameter of static graph $(\mathcal{N},\mathcal{E}_{\infty})$ where $\mathcal{E}_{\infty}:=\{(i,j)|(i,j)\in\mathcal{E}_k\text{ for infinitely many indices }k\}$.

Throughout this paper, we assume a ``intermittent''  connectivity condition of $\mathcal{G}_k$ as in Assumption \ref{assumption2}. It is a typical setting in asynchronous multi-agent optimization and control with a changing topology: the information propagation from one part to another is ensured during  a certain number of time slots, i.e., every $\overline{T}$ consecutive global time clocks. I will be used only for localized stopping criteria design later on.

\begin{assumption}[$\overline{T}$-strongly connected\cite{Nedic2014Distributed}]
	\label{assumption2}
	There exits a  $\overline{T}\in\mathbb{Z}^+$ such that the graph with edge set $\bigcup_{\tau=s\overline{T}}^{(s+1)\overline{T}-1}\mathcal{E}_{\tau}$ is strongly connected for every $s\in\mathbb{Z}^{+}$.
\end{assumption}

\subsection{Battery Model of PEVs}
\label{battery_model}
Consider there are $n$ PEVs to be charged over $T$ discrete time slots. Denote by $\mathcal{N}:=\left\{1,2,\cdots, n\right\}$ the set of PEVs, and $\mathcal T:=\{1,2,\cdots,T\}$ the set of time slots. For PEV processor $i$, let $p_{i}^{t}\in\mathbb{R}$ denote the charging power of PEV $i$ at time $t$ and we use $p_i$ to denote the column vector of $p_i^t$ over the entire time horizon of $\mathcal{T}$ for simplicity, i.e.,  $p_i:=(p_i^1, p_i^2, \ldots, p_i^T)^{\mathsf{T}}$.

Each PEV is available for load dispatch once it is plugged in and before the charging is completed. For arbitrary $i\in\mathcal{N}$, limits on the total charging amount should be satisfied according to its SOC, leading to constraints \eqref{c_limit}.
\begin{subequations}
	\label{c_limit}
	\begin{align}
	\label{cons1}
	&({ SOC}^{{ d}}_{i}-{ SOC}^{{ s}}_{i})\cdot { Cap}_i \leq \eta_{i} \cdot \sum\nolimits_{t\in\mathcal{T}} p_{i}^{t}  \\
	\label{cons2}
	&(\overline{{ SOC}_{i}}-{ SOC}^{{ s}}_{i})\cdot { Cap}_i \ge \eta_{i} \cdot \sum\nolimits_{t\in\mathcal{T}} p_{i}^{t}
	\end{align}
\end{subequations}
${ SOC}^{{ s}}_{i}$ and ${ SOC}^{{ d}}_{i}$ stand for the initial and final state of charge of PEV $i$, respectively. ${ Cap}_i$ is its battery capacity and $\eta_i$ scales the charging efficiency. Constraint (\ref{cons2}) implies that the charging process stops once the battery reaches its maximum state of charge (denoted by $\overline{{ SOC}}_i$). 

Each PEV can charge only after it plugs in at a certain time slot ${ T}^{ s} \in \mathcal{T}$ and before it leaves at  ${ T}^{ d} \in \mathcal{T}$, where ${ T}^{ s}<{ T}^{ d}$. Hence we have constraints (\ref{c_limit2}).
\begin{subequations}
	\label{c_limit2}
	\begin{align}
	\label{cons3}
	p_{i}^{t} \cdot (t-{ T}^{ s}_{i}) \ge 0,\quad t\in\mathcal{T}, \ i\in \mathcal{N}\\
	\label{cons4}
	p_{i}^{t} \cdot (t-{ T}^{ d}_{i}) \leq 0,\quad t\in\mathcal{T}, \ i\in \mathcal{N}
	\end{align}
\end{subequations}

At each time slot, the charging power of a PEV is assumed constant  but can vary from $ 0 $ to its maximum charging power  $\overline{{ p}}_i$ at different time slots. Then, we have the constraint (\ref{cons5}).
\begin{equation}
\label{cons5}
0 \leq p_{i}^{t} \leq \overline{{ p}}_{i}, \quad t\in \mathcal
T,\ i\in \mathcal{N}
\end{equation}

Combining (\ref{c_limit}) to (\ref{cons5}), we have individual feasible charging region over the time horizon of $\mathcal{T}$ denoted by $\mathcal{P}_i$, i.e.,
\begin{align}\label{feasible_region}
\mathcal{P}_{i}:=\{\ p_i\in\mathbb{R}^{T}\ |\ p_i\ \text{satisfies}\ \eqref{c_limit} - \eqref{cons5}\ \}.
\end{align}
Note that $\mathcal{P}_{i}$ should be kept to PEV processor $i$ itself, as the involved arrival time, departure time and battery state are private information.

\subsection{Coordination of  PEV Charging}
The charging coordination aims to minimize the total cost while satisfying  system operation constraints and individual charging demands. The problem can be formulated as follows. 
\begin{subequations}\label{COC}
	\begin{align}
	\label{con6a}
	{\rm CoC}:\quad &\min\nolimits_{}\ \sum\nolimits_{i\in\mathcal{N}} f_i(p_i)\\
	\label{con6b}
	{\rm s.t.}\quad&\sum\nolimits_{i\in\mathcal{N}} p_{i}\leq L\quad(\text{dual multiplier:}\ \pi)\\
	\label{con6c}
	\quad&p_i\in\mathcal{P}_{i},\forall i\in\mathcal{N}.
	\end{align}
\end{subequations}
where, $p_i$ is a $T$-dimensional decision vector representing the charging power of PEV $i$ in the $T$ time slots. The charging cost of PEV $i$, denoted by $f_i(p_i):\mathbb{R}^{T}\rightarrow\mathbb{R}$, is a convex function with respect to the charging power $p_i$.

Congestion due to feeder head capacity limit is considered in \eqref{con6b}. The right-hand-side parameter $L$ is a $T-$dimension vectorthat stands for the maximum available total charging power to avoid overload on feeder head at each time slot. $L$ is determined by the distribution system operator (DSO), and delivered to at least one processor. Note that $L$ can also be properly designed to achieve a valley-filling purpose or other demand response aims. Besides the coupling constraints, individual charging demand is captured by (\ref{con6c}) where $\mathcal{P}_i$ is processor $i$'s feasible region over the time horizon of $\mathcal{T}$ and its specific form is given by \eqref{feasible_region}. In practice, usually	$f_i(\cdot)$, $\mathcal{P}_i$ and $p_i$ are private information of user $i$, which can only be accessed by processor $i$ and should not be disclosed to others. 

For the convex problem CoC  \eqref{COC}, we also make the following regular assumption to guarantee the strong duality holds.

\begin{assumption}
	\label{Slater}
	The Slater's condition \cite[Chapter 5.2.3]{boyd2004convex}  holds for CoC, i.e., there exist $p_i$ in the interior of $\mathcal{P}_i$ such that (\ref{con6b}) holds.
\end{assumption}

%
%
%
%


\section{A Surrogate Model of CoC}
In this section, first we decompose \eqref{COC} by constructing an equivalent transformation, then derive a surrogate problem to mask the private information of individual PEVs during the iterative charging coordination, as we explain. 

First of all, invoking \cite[Chapter 5.1.1]{boyd2004convex}, the Lagrangian dual of (\ref{COC}) is given by
\begin{equation}
	\max_{\pi\ge0}\min_{
		p_i\in\mathcal{P}_i,\forall i
		}\left\{\pi^{\mathsf{T}}\left(-L+\sum\nolimits_{i\in\mathcal{N}}p_i\right)+\sum\nolimits_{i\in\mathcal{N}}f_i(p_i)\right\}\label{reform1}
\end{equation}
where  $\pi\in\mathbb{R}^{ T}$ is the dual variable vector corresponding to the global constraint \eqref{con6b}. (\ref{reform1}) can be further rewritten as
\begin{equation}
\label{reform2}
\begin{split}
\max_{\pi\ge0}\{\underbrace{\min_{p_{i^*}\in\mathcal{P}_{i^*}}\left\{f_{i^*}(p_{i^*})+\pi^{\mathsf{T}}(p_{i^*}-L)\right\}}_{\mathcal{U}_{i^*}(\pi)}\\
+\sum\nolimits_{i\in\mathcal{N}\backslash\left\{i^*\right\}}\underbrace{\min_{p_i\in\mathcal{P}_i}\left\{f_i(p_i)+\pi^{\mathsf{T}}p_i\right\}}_{\mathcal{U}_i(\pi)}\}.
\end{split}
\end{equation}
where $i^*$ who is informed of the value of $L$ is uniquely pre-determined. Now define local dual functions $\mathcal{U}_i(\pi):\mathbb{R}^{T}\mapsto\mathbb{R}$ as in (\ref{reform2}). Note that $-\mathcal{U}_i(\pi)$ are convex functions with respect to $\pi$\cite[Chapter 3.2.3]{boyd2004convex}. Then (\ref{reform2}) can be reformulated into (\ref{equivalent}), a convex problem, for succinctness. 
\begin{equation}
\max_{\pi\ge0}\sum\nolimits_{i\in\mathcal{N}}\mathcal{U}_i(\pi)\label{equivalent}
\end{equation}
The transformation from (\ref{COC}) to (\ref{equivalent}) is  built on the Lagrangian decomposition \cite[Chapter 4.3.1]{Bo2010Distributed}, by dualizing the coupling  constraint (\ref{con6b}) to obtain a separate structure as in (\ref{equivalent}). The equivalence is guaranteed by noting that strong duality holds under Assumption \ref{Slater}. Let $\pi^*$ denote the optimal solution of  problem (\ref{equivalent}). If $f_i(p_i)$ is strictly convex\footnote[1]{In case $f_i(p_i)$ is not strictly convex, say, it is linear in $p_i$, the difficulty in primal recovery 
 can be avoided by adding a sufficiently small quadratic item to the objective function without revising the optimal solution\cite{ref19.2}.}, the optimal solution of the primal CoC problem is uniquely determined by
\begin{equation}
\label{eq:back}
p_i^*=\arg\mathcal{U}_i(\pi^*),\forall i\in\mathcal{N}.
\end{equation}
Note that, without knowing other's $f_i(\cdot)$ and $\mathcal{P}_i$, one cannot infer other PEV's optimal charging profile, which protects the private information of individual PEVs.

Inspired by  Dantzig-Wolfe decomposition \cite{Dantzig}, we construct a further transformation on (\ref{equivalent}). New decision variables of the reformulated optimization problem, $z\in\mathbb{R}^{T+n}$, consists of two parts $\pi\in\mathbb{R}^{T}$ and $u:=col(u_i)\in\mathbb{R}^n$, where $u_i$ is introduced to replace the item $\mathcal{U}_i(\pi)$ in (\ref{equivalent}). The relationship between $\pi$ and $u_i$, which is originally described by $u_i=\mathcal{U}_i(\pi)$ as in (\ref{reform2}), is now captured by the feasible region of $z$ which is denoted by $\mathcal{Z}_i$ and identified in (\ref{define_Z}). Note that $Z_i$ is convex since $f_i$ is a convex function\cite[Chapter 3.2.3]{boyd2004convex}. Then,  problem (\ref{equivalent}) is equivalently converted into the convex problem (\ref{equivalent2}) where ${\rm e}:=(\bm{0}^{\mathsf{T}},\bm{1}^{\mathsf{T}})^{\mathsf{T}}$ is a constant  vector. 
\begin{subequations}
\label{equivalent20}	
\begin{align}
\label{equivalent2}
&\max_{z}{\rm e}^{\mathsf{T}}z,\ {\rm s.t.}\ z\in\cap_{i\in\mathcal{N}}\mathcal{Z}_i\\
\label{define_Z}
&\mathcal{Z}_i:=\left\{z\Bigg|
\begin{array}{l}
\pi\ge 0\\
u_i\le f_i(p_i)+\pi^{\mathsf{T}}p_i,\forall p_i\in\mathcal{P}_i,\text{if }i\in \mathcal{N}\backslash\left\{i^*\right\} \\
u_i\le f_i(p_i)+\pi^{\mathsf{T}}(p_i-L),\forall p_i\in\mathcal{P}_i,\text{if }i= i^*
\end{array}
\right\}
\end{align}
\end{subequations}
In this way, the original centralized optimization problem is converted into its surrogate dual model, the fully-distributed \eqref{equivalent2} with global decision variables $z$ and several isolated feasible regions of $z$. Note that  the dual function ${\rm e}^{\mathsf{T}}z$ as well as the set constraints $Z_i$ are both linear in $z$, providing great benefits for computation efficiency and enabling the cutting-plane exchange. In addition, as we will show in Section IV, this transformation enables processors to reach consensus with respect to the dual variable $z$ (or say $\pi$) under completely asynchronous center-free network without disclosing information about their local objective $f_i(\cdot)$, demand $\mathcal{P}_i$ and other sensitive data including $\mathcal{U}_i(\cdot)$ and $\mathcal{Z}_i$.

\section{Asynchronous Distributed Algorithm}

In this section, we develop a  cutting-plane based distributed algorithm to solve the surrogate optimization problem \eqref{equivalent20} asynchronously, and then derive an error-bounded solution under a completely asynchronous communication protocol. 

\subsection{Cutting-Plane based Asynchronous Distributed Algorithm}
Considering each PEV processor runs at its local clock $k_i$,  the asynchronous algorithm is derived as in Algorithm \ref{algorithm}. 

The basic idea of the algorithm \ref{algorithm} is as follows. A set of cutting planes are generated by each processor and is individually updated in every round of iterations. After reading the cutting-plane set from neighbors' output cache, the local processor collects all received cutting-planes and its own cutting-planes to form a polyhedron, denoted by $\mathcal{H}^{[i]}_{tmp}(k_i)$ for processor $i$ in the $k_i^{\rm th}$ round of local iteration\footnote[2]{In the rest of this paper, for a variable  $x^{[i]}_m(k)$, $i$ is processor $i$, $k$ is the $k^{\rm th}$ iteration round of processor $i$. $m$ stands for the $m^{\rm th}$ component of $x$}. The
polyhedron can be regarded as an approximation of the feasible region $\mathcal{S}:=\cap_{i\in\mathcal{N}}\mathcal{Z}_i$. Moreover, in each round of iteration, additional cutting planes are added to constantly shrink the polyhedron, leading to a more and more accurate estimation. Mathematically, this procedure is almost the same as the outer approximation method.  Since $\mathcal{H}^{[i]}(k_i)\subseteq\mathcal{H}^{[i]}(k_i-1)$ holds for each iteration,  $\mathcal{H}^{[i]}(k_i)$ will eventually  approach $\mathcal{S}$.  In this way, each local processor only needs to know its own $\mathcal{Z}_i$, and iteratively approaches the feasible region $\mathcal{S}$. 
Once the consensus on $\mathcal{S}$ is achieved, the consensus value of $z$ is obtained. Then each PEV processor can extract its optimal charging profile via (\ref{eq:back}).

For Algorithm \ref{algorithm}, we have the following useful remarks. 

\begin{algorithm}[h]
	\caption{Asynchronous PEV Charging for processor $i$}
	\footnotesize
	\label{algorithm}
	\hspace*{0.02in}{\bf Input:} 
	The local feasible region $\mathcal{P}_i$ of processor $i$\\
	\hspace*{0.02in}{\bf Iteration at} $k_i$: Suppose processor $i$'s clock ticks at $k_i$. Then it is activated to update its cutting-plane set as follows:\\
    \hspace*{0.02in}{\bf Step 1: Reading Phase}\\
    \hspace*{0.04in}
    Get cutting-plane set from its input-neighbors' output cache. Generate temporary cutting-plane set according to (\ref{alg_1}).\\
    \begin{equation}
    \label{alg_1}
    H^{[i]}_{tmp}(k_i)=\left(\cup_{j\in\mathcal{N}_I(i,k_i)}H^{[j]}(k_j)\right)\cup H^{[i]}(k_i)
    \end{equation}
	\hspace*{0.02in}{\bf Step 2: Computation Phase}\\
	\hspace*{0.04in}
	Solve a linear programming 
	\begin{align}
	z^{[i]}(k_i):=\arg\max\nolimits_{z} e^{\mathsf{T}}z-\rho{\Vert z\Vert}^2\ \text{s.t.}\ z\in\mathcal{H}^{[i]}_{tmp}(k_i)
	\end{align}
	where $\mathcal{H}^{[i]}_{tmp}(k_i)$ is the polyhedron induced by $H^{[i]}_{tmp}(k_i)$ and $\rho$ is a sufficiently small positive number. The penalty item $-\rho{\Vert z\Vert}^2$ is added to derive a unique solution which has minimal Euclidean norm. Then, shrink $H^{[i]}_{tmp}(k_i)$ by remaining the set of active constraints.\\
	\hspace*{0.04in} Based on $z^{[i]}(k_i)$ and $\mathcal{Z}_i$, generate a new cutting-plane which is denoted by $h_i(z^{[i]}(k_i))$ and its specific form is given in (\ref{gen_h}). Then, update local cutting-plane set according to (\ref{updateCCP}).\\
		\begin{equation}
		\label{updateCCP}
		H^{[i]}(k_i+1)=H^{[i]}_{tmp}(k_i)\cup h_i(z^{[i]}(k_i))
		\end{equation}
	\hspace*{0.02in}{\bf Step 3: Writing Phase}\\
	\hspace*{0.04in} Write $H^{[i]}(k_i+1)$ to its output cache. Update local clock by $k_i=k_i+1$.
\end{algorithm}

\begin{remark}[Asynchrony-resilience]
Note that no global clock is required in Algorithm \ref{algorithm}, implying intrinsic permission for asynchronous computation and updates. As demonstrated in Section \ref{case_study}, the proposed distributed algorithm is also resilient to other imperfect communication such as packet drops.
\end{remark}

\begin{remark}[Privacy-preserving]
The cutting-planes exchanged by processors are approximations of the feasible region $\mathcal{Z}_i$ of the surrogate problem (\ref{equivalent20}), other than the exact feasible region $\mathcal{P}_i$ or objective function $f_i(x)$ of the original CoC problem \eqref{COC}. Therefore, one's private  information will not be disclosed to  others. Moreover, since the union process in step 2 aggregates cutting-planes from  all processors, one cannot infer private information of any individual.

\end{remark}


\subsection{Distributed Generation of Cutting-planes}
The consensus on feasible region depends upon the generation of new cutting-planes. In this subsection, we omit $k_i$ for succinctness. Given the query point $z^{[i]}$ and a target set $\mathcal{Z}_i$, $h_i(z^{[i]})$ is generated as the cutting plane separating $z^{[i]}$ and $\mathcal{Z}_i$ if $z^{[i]}$ is not inside $\mathcal{Z}_i$. In order to identify if $z^{[i]}(k_i)$ is within $\mathcal{Z}_i$, processor $i$ needs to compare the values of $\mathcal{U}_i(\pi^{[i]})$ and $u_i^{[i]}$ according to the definition of $\mathcal{Z}_i$. Let $p_i^{[i]}:=\arg \mathcal{U}_i(\pi^{[i]})$ denote its optimal solution. Then processor $i$ generates the new cutting planes according to the specific from in (\ref{gen_h}).
\begin{align}
\label{gen_h}
h_i(z^{[i]}):=\left\{
\begin{array}{l}
h^{\emptyset},\ \text{if}\ u_i^{[i]}\le\mathcal{U}_i(\pi^{[i]});\\
\left\{z|u_i\le f_i(p_i^{[i]})+\pi^{\mathsf{T}}p_i^{[i]}\right\},\\
\quad\quad\text{if}\ u_i^{[i]}>\mathcal{U}_i(\pi^{[i]})\text{ and }i\in \mathcal{N}\backslash\left\{i^*\right\};\\
\left\{z|u_i\le f_i(p_i^{[i]})+\pi^{\mathsf{T}}(p_i^{[i]}-L)\right\},\\
\quad\quad\text{if}\ u_i^{[i]}>\mathcal{U}_i(\pi^{[i]})\text{ and }i=i*.
\end{array}
\right.
\end{align}

\subsection{Fully Distributed Initialization}
To develop the cutting-plane  based distributed algorithm, each local processor has to generate an initial cutting-plane set $H^{[i]}(0)$ without knowing the whole picture of the feasible region $\mathcal{S}$. To guarantee convergence, it is required that $\mathcal{S}\subset\mathcal{H}^{[i]}(0)$ and $\max_{z\in \mathcal{H}^{[i]}(0)} e^{\mathsf{T}}z<\infty$. To this end, we utilize the observation that the objective of \eqref{COC}, which represents a total costs of PEV charging, must have an upper bound in practice. Since strong duality holds for \eqref{COC} and \eqref{equivalent2}, there also exists an upper bound of the equivalent maximization problem (\ref{equivalent2}). Hence, each processor can individually choose a properly large number $M_i>0$ according to his historical data, and construct a initial cutting-plane set as
\begin{align}
\mathcal{H}^{[i]}(0)=\left\{z|e^{\mathsf{T}}z\le M_i\right\},i\in\mathcal{N}.
\end{align}

\subsection{Localized Stopping Criteria with Convergence Guarantee}
\label{sec_stopcri}
By implementing the proposed distributed algorithm, each local processor will derive a sequence of solutions during the iterations. It is crucial to find an appropriate stopping criteria for consensus. First we will introduce an empirical and centralized criterion, then extend it to a completely localized form. We will prove that the \emph{local criterion} is the sufficient condition for the \emph{global criterion}, deferring its detailed rationale, regarding optimality and feasibility, to subsection \ref{subsec_ConvergenceOptimality}. Before we start, 
the temporary objective value of processor $i$ at its $k_i$ round is denoted by
\begin{align}
\label{define_J}
J^{[i]}(k_i):&=e^{\mathsf{T}}z^{[i]}(k_i)-\rho{\Vert z^{[i]}(k_i)\Vert}^2.
\end{align}

\subsubsection{Global Criterion}
Denote the global objective error at $k$ by
\begin{align}
	\label{def_e}
	\max_{i,j\in\mathcal{N}} \vert J^{[i]}(k_i)-J^{[j]}(k_j)\vert
\end{align}
where $k_i,k_j$ are local clocks associated with the global clock $k$. Empirically, if (\ref{def_e}) is less than a pre-specified convergence tolerance, consensus on solution is regarded as been encountered and the algorithm terminates. The underlying rationale is that strict concavity of $J(\cdot)$ follows that $\vert J^{[i]}(k_i)-J^{[j]}(k_j)\vert\ge \sigma\Vert z^{[i]}(k_i)-z^{[j]}(k_j)\Vert^2$ for some $\sigma>0$\cite{ref17}. Therefore, the consensus on objective value can be approximated to that on solution $z$. Eq. \eqref{def_e}, however, is essentially a \emph{global criterion}, entailing temporary objective values from all local processors. It implies that individuals cannot implement this criterion by only accessing to local data. To circumvent this issue, a \emph{local criterion} is proposed below. 

\subsubsection{Local Criterion}
Given a pre-set tolerance $\epsilon>0$, two conditions constituting the \emph{local criterion} are given: 
\begin{condition}
\label{condition1}
For processor $i\in\mathcal{N}$, $J^{[i]}(k_i-K)-J^{[i]}\left(k_i\right)<\epsilon$ where $K:=d_{\mathcal{G}_{k}}\overline{T}$ is a constant with  $d_{\mathcal{G}_k}$ being the diameter of the communication topology $\mathcal{G}_k$ and $\overline{T}$ a parameter of $\mathcal{G}_k$ stated in Assumption \ref{assumption2}.
\end{condition}
\begin{condition}
\label{condition2}
For processor $i\in\mathcal{N}$, $u_i^{[i]}(k_i)-\mathcal{U}_i(\pi^{[i]}(k_i))<\epsilon$.
\end{condition}
Both Condition \ref{condition1} and \ref{condition2} are stated in localized form.  Condition \ref{condition1} claims to have 
stagnation on local objective updating within $K$ iterative steps\footnote[3]{This makes sense since $J^{[i]}(k_i)$ is monotonically nonincreasing with respect to $k_i$ as more constraints are added to the maximization problem while inactive constraints are pruned in every round of communication.}. Condition \ref{condition2} guarantees a bounded distance from the $z^{[i]}(k_i)$ in hand to $\mathcal{Z}_i$. The \emph{local criterion} is designed: processor $i$ stops at its local clock $k_i$ when Condition \ref{condition1} and \ref{condition2} are fulfilled with a pre-set tolerance $\epsilon>0$. 
 The \emph{local criterion} is  justified by theorem \ref{propo_criteria}.
\begin{theorem}
\label{propo_criteria}
If all processors have reached the \emph{local criterion} with Conditions \ref{condition1}-\ref{condition2} fulfilled, then the \emph{global criterion} is satisfied with $\max_{i,j\in\mathcal{N}} \vert J^{[i]}(k_i)-J^{[j]}(k_j)\vert<\epsilon$.
\end{theorem}

The proof of Theorem \ref{propo_criteria} is given in Appendix.\ref{proof_of_criteria}. Though the \emph{local criterion} may be more conservative than the global one, it only requires local information and enables a fully distributed and asynchronous implementation.

\subsection{Convergence and Optimality}
\label{subsec_ConvergenceOptimality}
Let $z^*$ and $J^*$ denote the optimal solution and optimal value of (\ref{equivalent20}) respectively. Convergence of Algorithm \ref{algorithm} is warranted by the monotonically nonincreasing objective value sequence $\{J^{[i]}(k_i)\}_{k_i}$, the lower bound of which is $J^*$. The optimality of Algorithm \ref{algorithm} is guaranteed by classic cutting-plane theory \cite{Eaves1971G,ref17}. Specifically, as the feasible region $\mathcal{S}$ is closed and compact, when Assumptions \ref{assumption2}-\ref{Slater} hold, the limit point of sequence $\{z^{[i]}(k_i)\}_{k_i}$ lies in $\mathcal{S}$, implying $J^*$ is greater than or equal to the limit of $\{J^{[i]}(k_i)\}_{k_i}$. Such being the case, the convergence and optimality of Algorithm \ref{algorithm} is ensured by
\begin{align}
\label{converg}
\lim\nolimits_{k_i\rightarrow \infty}J^{[i]}(k_i)=J^*,\forall i\in\mathcal{N}.
\end{align}

In practice, however, we are more concerned with the quality of the solutions obtained within finite rounds of iteration. Moreover, a consensus on $J$ may not necessarily imply that the optimal $J^*$ is achieved since $J^*$ is not knowable a priori for any individual processor (we will show by case studies in section \ref{case_study} how the \emph{global criterion} may fail at times). These points highlight the need for a measure to estimate the distance between a truncated solution $J^{[i]}(k_i)$ (or $z^{[i]}(k_i)$) and the exact solution $J^*$ (or $z^*$). Next we will show that, though $J^*$ (either $z^*$) does not appear in the Conditions \ref{condition1}-\ref{condition2}, the two conditions together guarantee bounded error of a truncated solution, with respect to optimality and feasibility.
\subsubsection{Optimality}
The optimality of the obtained solution can be characterized by the following theorem. 
\begin{theorem}[Optimality]
	\label{err_bound}
	Assume Conditions \ref{condition1}-\ref{condition2} hold for all processors in $\mathcal{N}$, then for any local processor $i\in\mathcal{N}$, $J^{[i]}(k_i)-J^*\ge 0$ and
	\begin{equation}
	\label{err_bound_epsilon}
	 J^{[i]}(k_i)-J^* \in O(\sqrt{\epsilon}).
	\end{equation}
\end{theorem}
\subsubsection{Feasibility}
There is no guaranteed feasibility of solution sequence in cutting-plane based algorithms. In other words, the obtainable consensus result within finite rounds may be very close to $\mathcal{S}$ but still  outside $\mathcal{S}$. To address this problem, we start with the situation that $z^{[i]}(k_i)$ is feasible, i.e., $z^{[i]}(k_i)\in\mathcal{S}$. Then we can easily infer from Theorem \ref{err_bound} that $J^{[i]}(k_i)=J^*$. This implies $z^{[i]}(k_i)$ must be the  optimal solution to (\ref{equivalent20}). If unfortunately  $z^{[i]}(k_i)$ is infeasible, it is revealed in Theorem \ref{pro_feasi} that $z^{[i]}(k_i)$ can be close enough to a feasible and quasi-optimal solution.
 
\begin{theorem}[Feasibility]
\label{pro_feasi}
Assume Conditions \ref{condition1}-\ref{condition2} hold for all processors in $\mathcal{N}$. Then for each local processor $i\in\mathcal{N}$, there exists a feasible solution $\bar{z}\in\mathcal{S}$ which is close enough to $z^{[i]}(k_i)$ with $\Vert z^{[i]}(k_i)-\bar{z} \Vert^2\le B\epsilon$ where $B$ is a positive constant, and nearly optimal with $\vert J(\bar{z})-J^*\vert\in O(\sqrt{\epsilon})$.
\end{theorem}

The proof of Theorem \ref{err_bound} and \ref{pro_feasi} is given in Appendix.B.

\section{Case Studies}
\label{case_study}
In this section, we test the performance of the proposed method and compare it with the celebrated ADMM algorithm. Simulations are carried out on the IEEE 37- and IEEE 123-node feeders\cite{ref20}, with MATLAB on a laptop with Intel(R) Core(TM) i5-5200U 2.20GHz CPU and 4GB of RAM.
\subsection{Setup}
\begin{figure}[!ht]
\centering
\includegraphics[scale=0.1]{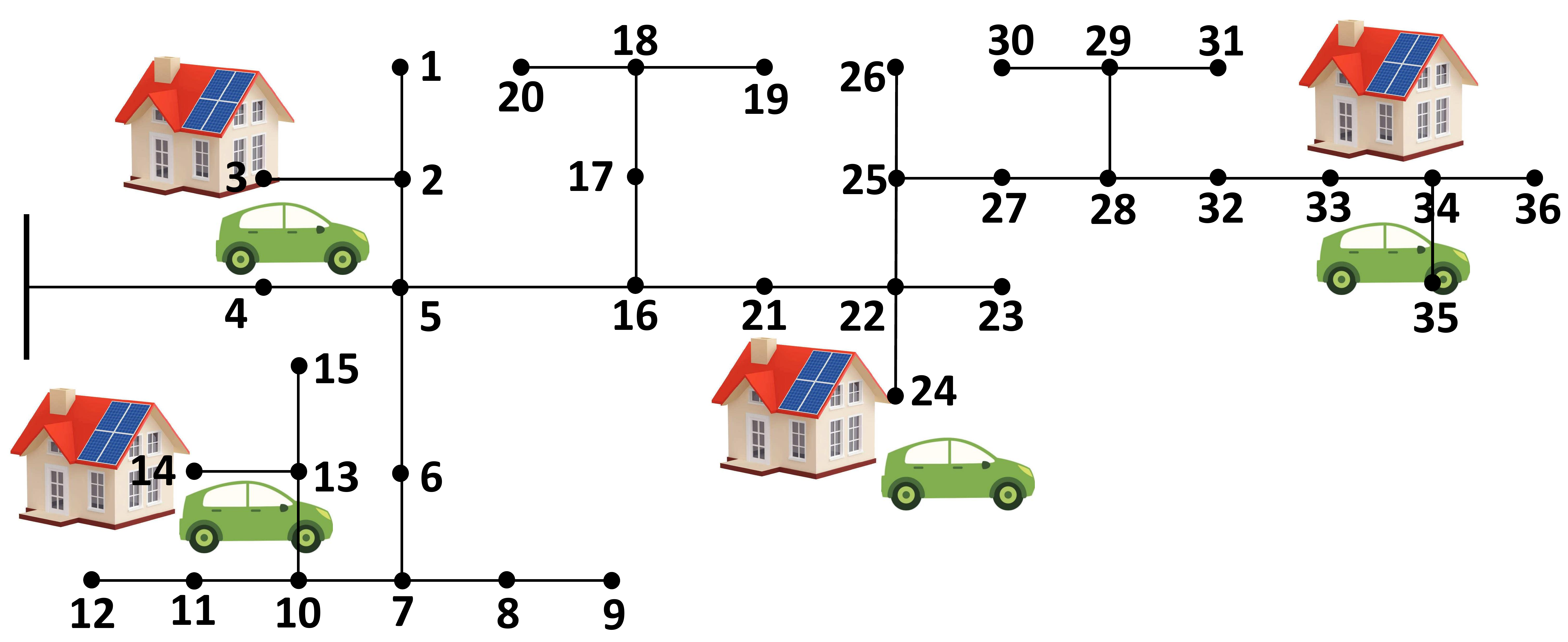}
\caption{Distributed network diagram of IEEE 37-node Test Feeder.}
\label{sysIEEE37}
\end{figure}
\textit{1) Physical and Communication Networks:} We consider two typical radial residential distribution networks: the IEEE 37-bus and 123-bus test feeders. The topology of the former system is depicted in Fig.\ref{sysIEEE37}, while the latter is omitted for space limitation. For both the systems, suppose that load of one household and one PEV is located at each bus. The feeder head (bus 4) is the only power supplier and its maximal capacity is set as the peak load without PEVs. We use the proposed method to coordinate PEVs' charging to avoid overload. The communication topology of charging PEVs is chosen similar to the power network topology for simplicity\footnote[4]{This setting is made solely for the clarity of presentation. Theoretically, the communication topology can be arbitrary provided Assumption \ref{assumption2} holds.}. We use real data from the hourly residential load profile of Los Angeles\cite{ref21} as the baseline load and scale it to match the household numbers. The information on hourly day-ahead electricity prices  comes from California ISO\cite{ref22}.

\textit{2) PEV Specifications:}
The  parameters of PEVs are given: Battery capacities lie in a uniform distribution between 18 kW.h to 20 kW.h\cite{ref23}. The scheduling horizon is from 5:00 pm to 9:00 am in the next day and is divided into 16 time slots hour-by-hour. Accordingly, we assume that the arrival and departure time of PEVs in the test cases lie in 5:00 pm-9:00 am and their hourly probability distributions are determined according to \cite{ref24}. 
Initial and designated SOC are uniformly distributed in $[0.3,0.5]$ and $[0.7,0.9]$ respectively\cite{Liu2017Decentralized}. The maximum charge power is set as 3.3 kW for Level \RNum{2} charger. A charging efficiency of 0.9 is  considered.

\subsection{Optimality}
In this case, Algorithm \ref{algorithm} is applied to the IEEE 37-node test feeder. Results are presented in Fig.\ref{net37load}, compared with  two common uncoordinated charging modes:
\begin{itemize}
\item \emph{mode (\romannumeral1)}: PEVs start charging immediately when they arrive and until the designated SOC is reached.
\item \emph{mode (\romannumeral2)}: PEVs optimize charging cost on their own without coordinating with others.
\end{itemize}
A total charging cost of \$109.12, \$130.60 and \$105.06 is achieved under coordinated charging, uncoordinated mode (\romannumeral1) and (\romannumeral2) respectively. As shown in Fig.\ref{net37load}, under mode (\romannumeral1), a new peak load is imposed on the baseload profile around evening rush hours, which gives rise to a great burden on the feeder head. Moreover, mode (\romannumeral1) costs the most because of charging during high price periods, which is uneconomical. As for mode (\romannumeral2), a minimal charging cost is achieved, however, with there being a new peak load around 2 am to 3 am. The aggregation charging behavior during low price period around midnight also threatens the system security, although it is in off-peak time for non-EV baseload. The proposed coordinated charging meets system constraints at minimum cost, striving a balance between security and economic efficiency. It tries to schedule PEV load to off-peak time as well as avoids overload on the feeder head. An observation is made that the strategy plays a role in `valley-filling' of total load profile as a consequence when PEV load and baseload are roughly on the same scale. The outcomes of the distributed algorithm are coincident with that of the centralized method, which is solved by  CPLEX.
\begin{figure}[!ht]
	\centering
	\includegraphics[width=0.32\textwidth, height=1.5in]{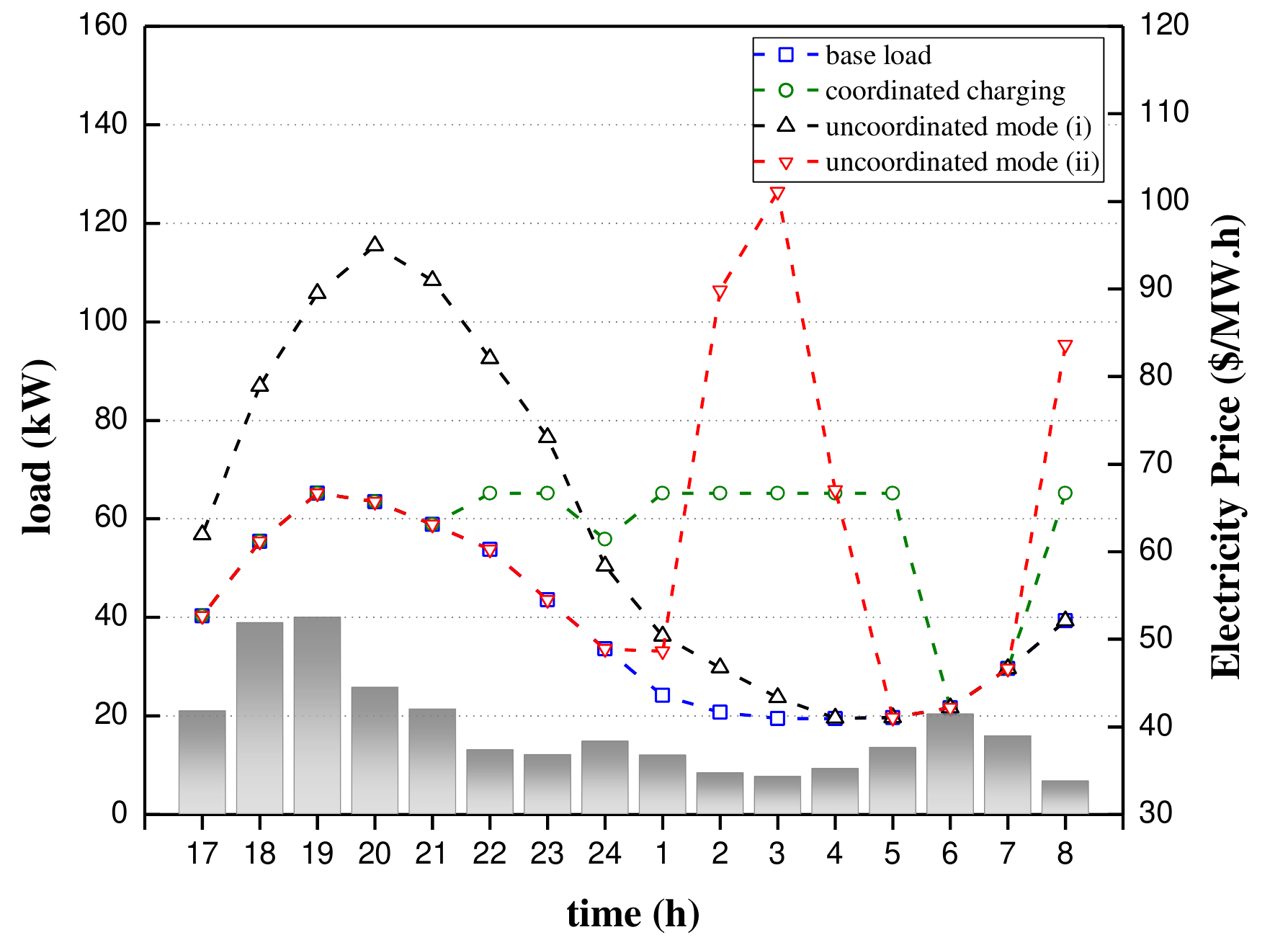}
	\caption{Total load profile under coordinated/uncoordinated charging modes.}
	\label{net37load}
\end{figure}
\subsection{Convergence}
Fig.\ref{convergence}(left) shows the iterative process of several selected nodes, converging to the optimal $J^*$. An observation is made on the different initial values of the optimization objective for individual processors. This is because they independently choose their $M_i$ uniformly in $[150,200]$ to generate their initial cutting plane set. Also, it is observed that $J^{[i]}(k)$ of each processor is monotonously nonincreasing along with the communication round $k$, according with theoretical analysis.

To demonstrate performance of the different criterion, we show the evolution of four alternative errors: (\romannumeral1) $e_{\text{\uppercase\expandafter{\romannumeral1}}}(k)=\max_{i}\vert J^{[i]}(k)-J^*\vert$; (\romannumeral2) $e_{\text{\uppercase\expandafter{\romannumeral2}}}(k)$ which is exactly the \emph{global criterion} in (\ref{def_e}); (\romannumeral3) $e_{\text{\uppercase\expandafter{\romannumeral3}}}(k)=\max_{i}\{J^{[i]}(k-K)-J^{[i]}(k)\}$ where $K=15$ in this case, and (\romannumeral4) $e_{\text{\uppercase\expandafter{\romannumeral4}}}(k)=\max_{i}\{u_i^{[i]}(k)-\mathcal{U}_i(\pi^{[i]}(k))\}$, as in Fig.\ref{convergence}(right). $e_{\text{\uppercase\expandafter{\romannumeral1}}}(k)$ essentially characterizes the convergence performance of the algorithm, however, entails the optimal $J^*$ being known a priori. Note that the $e_{\text{\uppercase\expandafter{\romannumeral3}}}(k)$ associated with Condition.\ref{condition1}
and $e_{\text{\uppercase\expandafter{\romannumeral4}}}(k)$ associated with Condition.\ref{condition2} together give a sketch of the proposed \emph{local criterion}. As shown in Fig.\ref{convergence}(right), the algorithm converges after about 18 rounds of iteration, with $e_{\text{\uppercase\expandafter{\romannumeral1}}}(k)$ reduced to less than $10^{-3}$. Notice that the \emph{global criterion} fails in this case, since the $e_{\text{\uppercase\expandafter{\romannumeral2}}}(k)$ reduces almost to $0$ (the shaded box in Fig.\ref{convergence}(right)) whereas achieving consensus on a non-optimal objective value. The \emph{local criterion} with $\epsilon=10^{-3}$ is met at round 33 with $e_{\text{\uppercase\expandafter{\romannumeral3}}}(k)<\epsilon$ (namely, Condition.\ref{condition1} satisfied by all processors) and $e_{\text{\uppercase\expandafter{\romannumeral4}}}(k)<\epsilon$ (namely, Condition.\ref{condition2} fulfilled by all). Though the \emph{local criterion} may be conservative due to Condition.\ref{condition1}, it provides guaranteed optimality of the consensus outturn, as opposed to the empirical \emph{global criterion}.

\begin{figure}[!ht]
	\centering
	\includegraphics[height=1.45in]{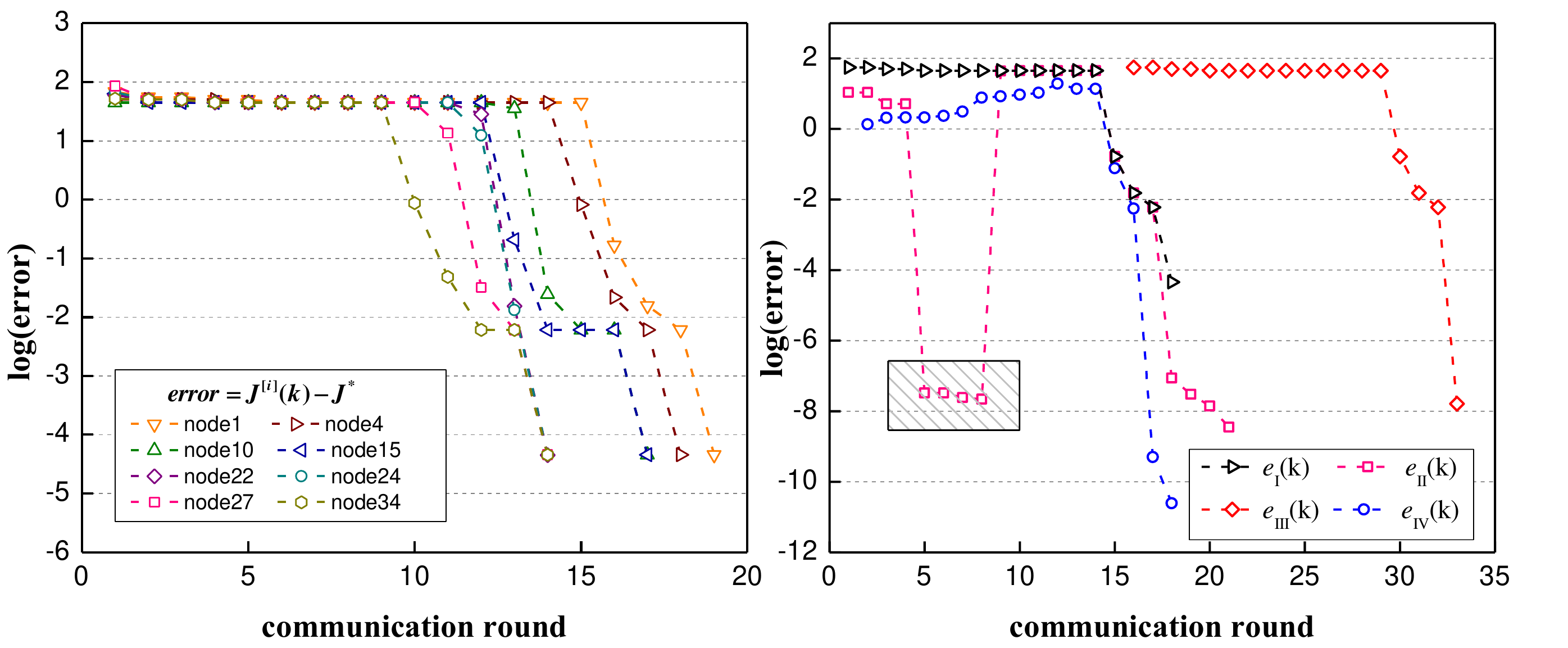}
	\caption{Evolution of the algorithm (infinitely negative values are omitted).}
	\label{convergence}
\end{figure}
We adopt the celebrated ADMM algorithm
in \cite{distributedADMM2}, which is also center-free, for comparison. Also, case of IEEE 123-node feeder is tested to showcase the scalability to large problems. The error tolerance is set as $\epsilon=10^{-3}$. Table \RNum{2} compares the convergence performances of the two methods. It is observed that the proposed algorithm  needs only  half of the communication rounds of ADMM to converge, showing a better convergence performance. 
\begin{table}[!h]
	\centering
	\caption{Contrast of Communication Rounds ($\epsilon=10^{-3}$)}
	\begin{tabular}{lcccc}
		\toprule
		\multirow{2}{*}{Cases}       &\multicolumn{3}{c}{The Proposed Method}     &ADMM\\
		               &$e_{\text{\uppercase\expandafter{\romannumeral1}}}<\epsilon$ & $e_{\text{\uppercase\expandafter{\romannumeral3}}}<\epsilon$ (Cond.\ref{condition1})& $e_{\text{\uppercase\expandafter{\romannumeral4}}}<\epsilon$ (Cond.\ref{condition2})& $e_{\text{\uppercase\expandafter{\romannumeral1}}}<\epsilon$\\
		\midrule
		IEEE37        &18&33&16               &51\\
		IEEE123       &42&66&42               &113\\
		\bottomrule
	\end{tabular}
\end{table}

\subsection{Performance of Asynchronous Charging}
In this subsection, the robustness of the algorithm to communication delay, packet loss and communication topology changes are tested in IEEE 37-node feeder.

\subsubsection{Communication Delay and  Packet Losses}
We assume that in each time step, each communication link has a delay of one time step with a probability of 10$\%$. Also we suppose that the packet loss probability is 10$\%$ for  the transmitted cutting-planes. Results  given  in Fig.\ref{results_delay}(left) evidently illustrate a satisfactory performance of the proposed algorithm even with communication delay and package losses.

\begin{figure}[!ht]
	\centering
	\includegraphics[height=1.35in]{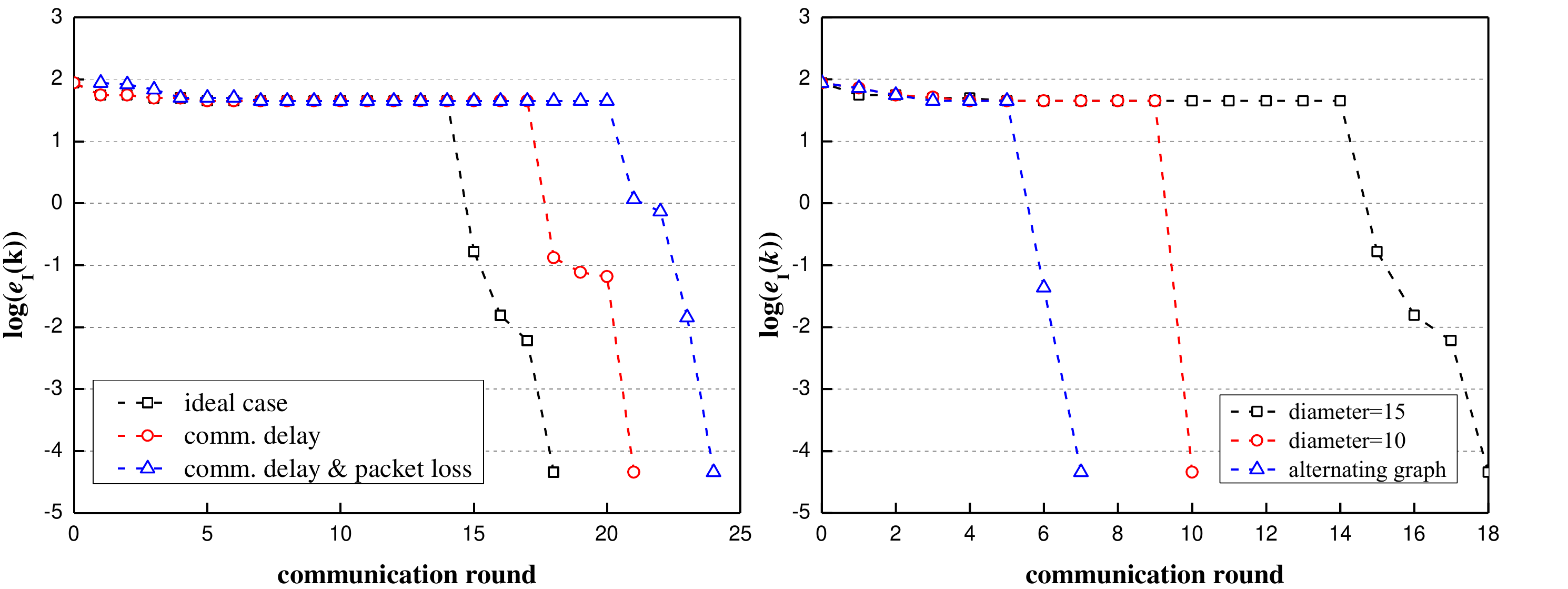}
	\caption{Performance of the algorithm characterized by $e_{\text{I}}$ under: (left) comm. delay and packet loss, and (right) different comm. topologies.}
	\label{results_delay}
\end{figure}

\subsubsection{Topology  Varying of Communication Network }
 The original IEEE 37-node feeder topology has a diameter of 15. Consider another topology with a diameter of 10. Cases with different communication topologies are tested and the comparison of convergence performances are shown in Fig.\ref{results_delay}(right). It is observed that a static graph with bigger diameter calls for more iteration rounds to converge. The rationale for this observation is that the diameter of the communication topology determines the longest time needed for passing cutting-planes from one node to another indirectly. Surprisingly, when alternating the two topologies in coordinated charging (with Assumption.\ref{assumption2} hold always), the algorithm converges even faster than the cases with static topologies. This fact indicates that, with the proposed algorithm, topology varying may even accelerate the cutting-plane passing process in the network, which could  facilitate the convergence.

\subsection{Plug-and-Play Operation}
The proposed algorithm is tested on the IEEE 37-node system in a plug-and-play operation. Classify the 36 nodes into two parts $\mathcal{N}_1=\left\{1,2,\ldots,20\right\}$ and $\mathcal{N}_2=\left\{21,22,\ldots,36\right\}$. At the beginning, only PEVs in $\mathcal{N}_1$ coordinate on charging. Nodes in $\mathcal{N}_2$ participate at the $16$ round. Results are shown in Fig.\ref{results_change}, showing that new players can join in at any time, which well supports the plug-and-play operation.
\begin{figure}[!ht]
	\centering
	\includegraphics[height=1.5in]{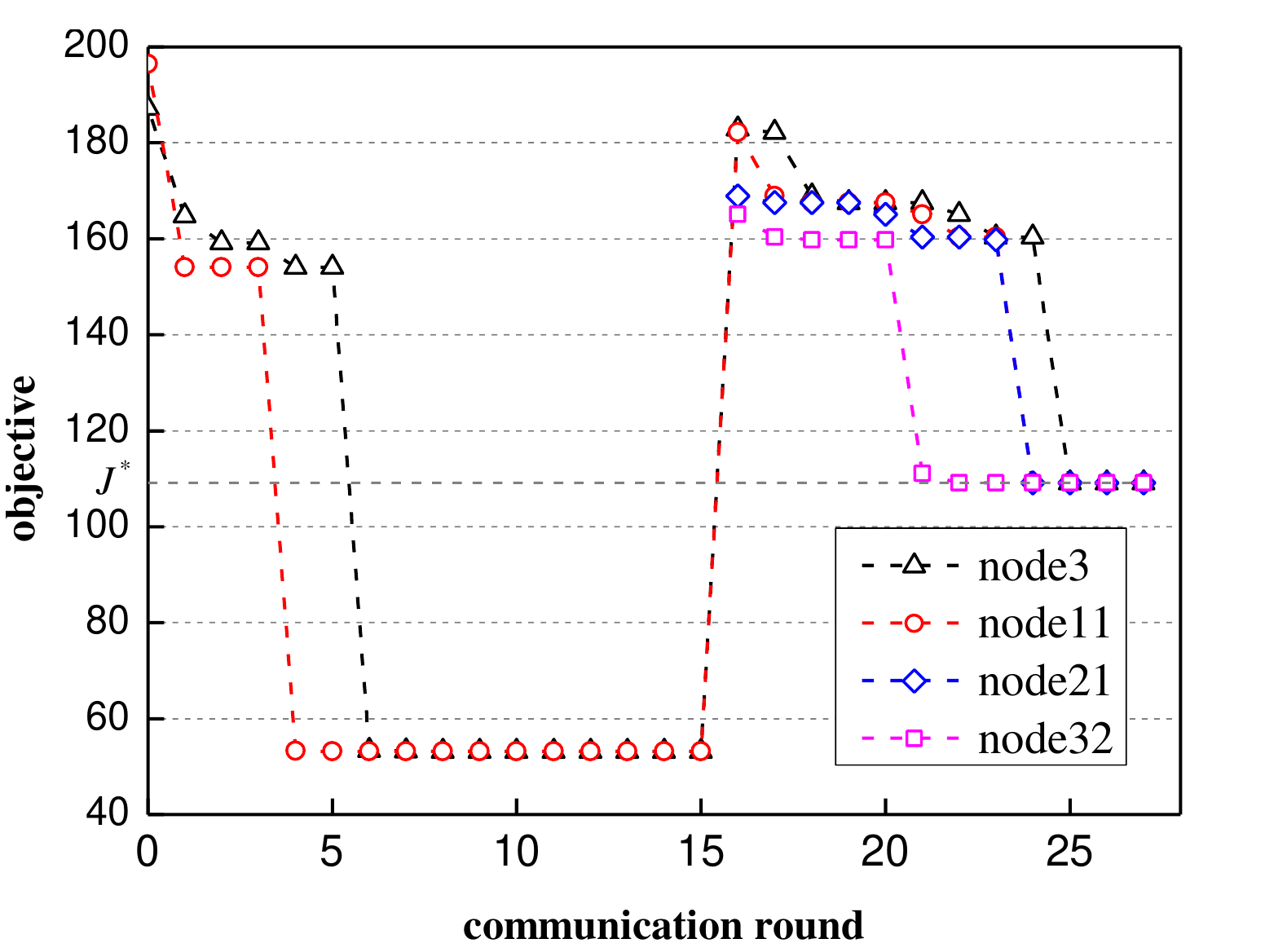}
	\caption{Evolution of local objective value of selected nodes under plug and play setting.}
	\label{results_change}
\end{figure}
\section{Conclusion}
In this paper, we have derived a cutting-plane based method to fulfill an optimal distributed coordination of PEV charging under local and global constraints. The proposed method strives the minimal overall charging cost without violating  feeder head capacity, which is in accordance to  the  result of  centralized global optimization.  During the  PEV charging, private information of individual PEVs can be well protected. It  performs resiliently under various kinds of asynchrony in practice such as time delays, packet losses, and  topology changes. 
We hope this work can promote an alternative path toward a more practical protocol for PEV charging and other distributed coordination problems.

As an initial study, uncertainties of renewable resources and inaccuracy of baseload  forecast have not been taken into account in this work. Extending the proposed cutting-plane based distributed coordination framework to incorporate uncertainties are among our ongoing  works.

\ifCLASSOPTIONcaptionsoff
  \newpage
\fi
\bibliographystyle{IEEEtran}
\bibliography{IEEEabrv,mybib}

\begin{thebibliography}{10}
\providecommand{\url}[1]{#1}
\csname url@samestyle\endcsname
\providecommand{\newblock}{\relax}
\providecommand{\bibinfo}[2]{#2}
\providecommand{\BIBentrySTDinterwordspacing}{\spaceskip=0pt\relax}
\providecommand{\BIBentryALTinterwordstretchfactor}{4}
\providecommand{\BIBentryALTinterwordspacing}{\spaceskip=\fontdimen2\font plus
\BIBentryALTinterwordstretchfactor\fontdimen3\font minus
  \fontdimen4\font\relax}
\providecommand{\BIBforeignlanguage}[2]{{%
\expandafter\ifx\csname l@#1\endcsname\relax
\typeout{** WARNING: IEEEtran.bst: No hyphenation pattern has been}%
\typeout{** loaded for the language `#1'. Using the pattern for}%
\typeout{** the default language instead.}%
\else
\language=\csname l@#1\endcsname
\fi
#2}}
\providecommand{\BIBdecl}{\relax}
\BIBdecl

\bibitem{Gan2013Optimal}
L.~Gan, U.~Topcu, and S.~H. Low, ``Optimal decentralized protocol for electric
  vehicle charging,'' \emph{IEEE Transactions on Power Systems}, vol.~28,
  no.~2, pp. 940--951, 2013.

\bibitem{ref3}
J.~A.~P. Lopes, F.~J. Soares, and P.~M.~R. Almeida, ``Integration of electric
  vehicles in the electric power system,'' \emph{Proceedings of the IEEE},
  vol.~99, no.~1, pp. 168--183, 2011.

\bibitem{Engel2013Privacy}
D.~Engel, ``Privacy and security challenges in the smart grid user domain,'' in
  \emph{Proceedings of the First ACM Workshop on Information Hiding \&
  Multimedia Security}, 2013.

\bibitem{Omran2017A}
N.~G. Omran and S.~Filizadeh, ``A semi-cooperative decentralized scheduling
  scheme for plug-in electric vehicle charging demand,'' \emph{International
  Journal of Electrical Power \& Energy Systems}, vol.~88, no. Complete, pp.
  119--132, 2017.

\bibitem{Turitsyn2010Robust}
K.~Turitsyn, N.~Sinitsyn, S.~Backhaus, and M.~Chertkov, ``Robust
  broadcast-communication control of electric vehicle charging,'' in
  \emph{First IEEE International Conference on Smart Grid Communications},
  2010.

\bibitem{Lopes2009Identifying}
J.~A.~P. Lopes, F.~J. Soares, and P.~M.~R. Almeida, ``Identifying management
  procedures to deal with connection of electric vehicles in the grid,'' in
  \emph{IEEE Bucharest PowerTech}, 2009.

\bibitem{Vaya2012Centralized}
M.~G. Vaya and G.~Andersson, ``Centralized and decentralized approaches to
  smart charging of plug-in vehicles,'' in \emph{Power \& Energy Society
  General Meeting}, 2012.

\bibitem{Jiang2013Decentralized}
B.~Jiang and Y.~Fei, ``Decentralized scheduling of pev on-street parking and
  charging for smart grid reactive power compensation,'' in \emph{IEEE PES
  ISGT}, 2013.

\bibitem{Liu2017Decentralized}
M.~Liu, P.~K. Phanivong, S.~Yang, and D.~S. Callaway, ``Decentralized charging
  control of electric vehicles in residential distribution networks,''
  \emph{IEEE Transactions on Control Systems Technology}, vol.~27, no.~1, pp.
  266--281, 2019.

\bibitem{Ardakanian2013Distributed}
O.~Ardakanian, C.~Rosenberg, and S.~Keshav, ``Distributed control of electric
  vehicle charging,'' in \emph{International Conference on Future Energy
  Systems}, 2013.

\bibitem{Liang2016Scalable}
L.~Zhang, V.~Kekatos, and G.~B. Giannakis, ``Scalable electric vehicle charging
  protocols,'' \emph{IEEE Transactions on Power Systems}, vol.~32, no.~2, pp.
  1451--1462, 2017.

\bibitem{Vaya2015Decentralized}
M.~G. Vaya, G.~Andersson, and S.~Boyd, ``Decentralized control of plug-in
  electric vehicles under driving uncertainty,'' in \emph{IEEE PES ISGT
  Europe}, 2014.

\bibitem{Kraning2014Dynamic}
M.~Kraning, ``Dynamic network energy management via proximal message passing,''
  \emph{Foundations \& Trends in Optimization}, vol.~1, no.~2, pp. 73--126,
  2014.

\bibitem{Rivera2013Alternating}
J.~Rivera, P.~Wolfrum, S.~Hirche, C.~Goebel, and H.~A. Jacobsen, ``Alternating
  direction method of multipliers for decentralized electric vehicle charging
  control,'' in \emph{52nd IEEE Conference on Decision \& Control}, 2013.

\bibitem{Ma2012Decentralized}
Z.~Ma, D.~S. Callaway, and I.~A. Hiskens, ``Decentralized charging control of
  large populations of plug-in electric vehicles,'' \emph{IEEE Transactions on
  Control Systems Technology}, vol.~21, no.~1, pp. 67--78, 2012.

\bibitem{Parise2014Mean}
F.~Parise, M.~Colombino, S.~Grammatico, and J.~Lygeros, ``Mean field
  constrained charging policy for large populations of plug-in electric
  vehicles,'' in \emph{53rd IEEE Conference on Decision \& Control}, 2014.

\bibitem{Logenthiran2012Multi}
T.~Logenthiran, D.~Srinivasan, A.~M. Khambadkone, and H.~N. Aung, ``Multi-agent
  system for real-time operation of a microgrid in real-time digital
  simulator,'' \emph{IEEE Transactions on Smart Grid}, vol.~3, no.~2, pp.
  925--933, 2012.

\bibitem{MohammadiA}
J.~Mohammadi, G.~Hug, and S.~Kar, ``A fully distributed cooperative charging
  approach for plug-in electric vehicles,'' \emph{IEEE Transactions on Smart
  Grid}, vol.~9, no.~4, pp. 3507--3518, 2018.

\bibitem{Rahbari2014Cooperative}
N.~Rahbari-Asr and M.~Y. Chow, ``Cooperative distributed demand management for
  community charging of phev/pevs based on kkt conditions and consensus
  networks,'' \emph{IEEE Transactions on Industrial Informatics}, vol.~10,
  no.~3, pp. 1907--1916, 2014.

\bibitem{Xu2015Optimal}
Y.~Xu, ``Optimal distributed charging rate control of plug-in electric vehicles
  for demand management,'' \emph{IEEE Transactions on Power Systems}, vol.~30,
  no.~3, pp. 1536--1545, 2015.

\bibitem{Low1999}
S.~H. {Low} and D.~E. {Lapsley}, ``Optimization flow control. i. basic
  algorithm and convergence,'' \emph{IEEE/ACM Transactions on Networking},
  vol.~7, no.~6, pp. 861--874, 1999.

\bibitem{Wang2018D}
Z.~Wang, S.~Mei, L.~Feng, P.~Yi, and M.~Cao, ``Asynchronous distributed power
  control of multi-microgrid systems,'' \emph{arXiv preprint arXiv:1810.11998},
  2018.

\bibitem{HaleAsynchronous}
M.~T. Hale, A.~Nedich, and M.~Egerstedt, ``Asynchronous multi-agent primal-dual
  optimization,'' \emph{IEEE Transactions on Automatic Control}, vol.~62,
  no.~9, pp. 4431--4435, 2017.

\bibitem{Eaves1971G}
B.~Eaves and W.~Zangwill, ``Generalized cutting plane algorithms,'' \emph{SIAM
  Journal on Control}, vol.~9, no.~4, pp. 529--542, 1971.

\bibitem{Liu2016Fully}
M.~Liu, J.~Zhu, L.~Li, and W.~Zhao, ``Fully decentralized multi-area dynamic
  economic dispatch for large-scale power systems via cutting plane
  consensus,'' \emph{IET Generation Transmission \& Distribution}, vol.~10,
  no.~10, pp. 2486--2495, 2016.

\bibitem{ref17}
M.~B\"urger, G.~Notarstefano, and F.~Allg\"ower, ``A polyhedral approximation
  framework for convex and robust distributed optimization,'' \emph{IEEE
  Transactions on Automatic Control}, vol.~59, no.~2, pp. 384--395, 2014.

\bibitem{distributedADMM2}
E.~{Wei} and A.~{Ozdaglar}, ``Distributed alternating direction method of
  multipliers,'' in \emph{51st IEEE Conference on Decision and Control}, 2012.

\bibitem{Nedic2014Distributed}
A.~Nedic and A.~Olshevsky, ``Distributed optimization over time-varying
  directed graphs,'' \emph{IEEE Transactions on Automatic Control}, vol.~60,
  no.~3, pp. 601--615, 2014.

\bibitem{boyd2004convex}
S.~Boyd and L.~Vandenberghe, \emph{Convex optimization}.\hskip 1em plus 0.5em
  minus 0.4em\relax Cambridge university press, 2004.

\bibitem{Bo2010Distributed}
Y.~Bo and M.~Johansson, ``Distributed optimization and games: A tutorial
  overview,'' \emph{Networked Control Systems}, vol. 406, pp. 109--148, 2010.

\bibitem{ref19.2}
O.~L. Mangasarian and R.~R. Meyer, ``Nonlinear perturbation of linear
  programs,'' \emph{Siam Journal on Control \& Optimization}, vol.~17, no.~6,
  pp. 745--752, 1978.

\bibitem{Dantzig}
G.~B.~Dantzig and P.~Wolfe, ``The decomposition algorithm for linear
  programs,'' \emph{Econometrica}, vol.~29, no.~4, pp. 767--778, 1961.

\bibitem{ref20}
\BIBentryALTinterwordspacing
{IEEE PES AMPS DSAS Test Feeder Working Group}. [Online]. Available:
  \url{{http://sites.ieee.org/pes-testfeeders/resources/}}
\BIBentrySTDinterwordspacing

\bibitem{ref21}
\BIBentryALTinterwordspacing
{Open EI Datasets}. [Online]. Available:
  \url{{https://openei.org/datasets/dataset/}}
\BIBentrySTDinterwordspacing

\bibitem{ref22}
\BIBentryALTinterwordspacing
{Energy Online}. [Online]. Available: \url{{http://www.energyonline.com/Data/}}
\BIBentrySTDinterwordspacing

\bibitem{ref23}
\BIBentryALTinterwordspacing
{UEP Agency}. [Online]. Available:
  \url{{http://www.fueleconomy.gov/feg/evsbs.shtml}}
\BIBentrySTDinterwordspacing

\bibitem{ref24}
N.~G. Omran and S.~Filizadeh, ``A semi-cooperative decentralized scheduling
  scheme for plug-in electric vehicle charging demand,'' \emph{International
  Journal of Electrical Power \& Energy Systems}, vol.~88, no. Complete, pp.
  119--132, 2017.

\end{thebibliography}


%
%
%



\newpage

\appendices

\section{Proof of the Theorem \ref{propo_criteria}}
\label{proof_of_criteria}
Let $t^{[i]}(k_i)$ denote the mapping from processor $i$'s local clock $k_i$ to its corresponding unique universal time and $t(k)$ denote the mapping from global clock $k$ to its corresponding unique universal time. $\forall i,k$, we define
\begin{equation}
Q^{[i]}(k):=J^{[i]}(k_i)
\end{equation}
where $t^{[i]}(k_i)\le t(k)< t^{[i]}(k_i+1)$. Since $J^{[i]}(k_i)$ is monotonously non-increasing with respect to $k_i$, $Q^{[i]}(k)$ is also monotonously non-increasing with respect to $k$. 
\begin{lemma}
	\label{B_lemma}
	For any $i,k_i$ and $k$, if $t^{[i]}(k_i)\le t(k)< t^{[i]}(k_i+1)$, then $J^{[i]}(k_i-K)\ge Q^{[i]}(k-K)$ where $K\in \mathbb{Z}^{+}$ is an arbitrary constant.
\end{lemma}

\begin{proof}
	Suppose $Q^{[i]}(k-K)=J^{[i]}(\tilde{k_i})$ where
	\begin{equation}
	\label{inequality1}
	t^{[i]}(\tilde{k_i})\le t(k-K)< t^{[i]}(\tilde{k_i}+1).
	\end{equation}
	Due to the relationship of global clock and local clocks, if $t^{[i]}(k_i)\le t(k)< t^{[i]}(k_i+1)$, there must be
	\begin{equation}
	\label{inequality2}
	t^{[i]}(k_i-K)\le t(k-K).
	\end{equation}
	Combining (\ref{inequality2}) and (\ref{inequality1}), we have
	\begin{equation}
	t^{[i]}(k_i-K)< t^{[i]}(\tilde{k_i}+1)
	\end{equation}
	which is equivalent to 
	\begin{equation}
	k_i-K\le \tilde{k_i}.
	\end{equation}
	Since $J^{[i]}(k_i)$ is monotonously non-increasing with respect to $k_i$, we have 
	\begin{equation}
	J^{[i]}(k_i-K)\ge J^{[i]}(\tilde{k_i})=Q^{[i]}(k-K)
	\end{equation}
	which completes the proof of lemma \ref{B_lemma}.
\end{proof}

Now proof of theorem \ref{propo_criteria} is given as follows by contradiction:
\begin{proof}
	We have assumed that $\forall i,j$
	\begin{subequations}
		\begin{align}
		\label{B_given_assumption1}
		&J^{[i]}(k_i-K)-J^{[i]}(k_i)<\epsilon,\\
		\label{B_given_assumption2}
		&J^{[j]}(k_j-K)-J^{[j]}(k_j)<\epsilon.
		\end{align}
	\end{subequations}
	Now suppose there exists $i,j$ and $k$ such that:
	\begin{subequations}
		\begin{align}
		\label{B_assmp1}
		&Q^{[j]}(k)-Q^{[i]}(k)>\epsilon,\\
		\label{B_assmp2}
		&t^{[i]}(k_i)\le t(k)< t^{[i]}(k_i+1),\\
		&t^{[j]}(k_j)\le t(k)< t^{[j]}(k_j+1).
		\end{align}
	\end{subequations}
	According to (\ref{B_assmp2}) and the definition of $Q^{[\cdot]}(\cdot)$, we have 
	\begin{equation}
	\label{B_equality1}
	Q^{[i]}(k)=J^{[i]}(k_i).
	\end{equation}
	According to (\ref{B_assmp2}) and lemma \ref{B_lemma}, we have 
	\begin{equation}
	\label{B_equality2}
	Q^{[i]}(k-K)\le J^{[i]}(k_i-K).
	\end{equation}
	Then combining (\ref{B_given_assumption1}), (\ref{B_equality1}) and (\ref{B_equality2}), we have
	\begin{equation}
	\label{B_inequality3}
	Q^{[i]}(k-K)<Q^{[i]}(k)+\epsilon.
	\end{equation}
	Combining (\ref{B_assmp1}) and (\ref{B_inequality3}) we derive:
	\begin{equation}
	\label{B_contradiction}
	Q^{[i]}(k-K)<Q^{[j]}(k)
	\end{equation}
	
	Denote $k_{0}:=k-K$ and index set $I_0=\left\{i\right\}$. For any $s>0$,
	\begin{equation}
	I_s:=\left\{m:\exists l\in I_{s-1},\ {\rm s.t.}(m,l)\in \mathcal{E}_{k_0+s} \right\}
	\end{equation}
	As the communication graph is strongly connected with limited intercommunication interval, $I_{K}$ must contains all the nodes in the network. Therefore there exists $1\le s^{*}\le K$ such that $j\in I_{s^*}$. The algorithm guarantees that
\begin{equation}
\forall l\in I_s,Q^{[l]}(k_0+s)\le Q^{[i]}(k_0).
\end{equation}
Thus we have
\begin{equation}
Q^{[j]}(k_0+K)\le Q^{[j]}(k_0+s^*)\le Q^{[i]}(k_0)
\end{equation}
which contradicts (\ref{B_contradiction}).
\end{proof}

\section{Proof of Theorem \ref{err_bound} and \ref{pro_feasi}}
\label{proof_err_bound}
According to the model in Appendix \ref{battery_model}, for any $i\in\mathcal{N}$, $p_i$ is bounded in set $\mathcal{P}_i$ with $\Vert p_i \Vert\le \sqrt{T}\overline{p}_i$. Before giving the proof of theorem \ref{err_bound} and \ref{pro_feasi}, two lemmas are provided as follows.
\begin{lemma}
\label{C_lemma1}
If Conditions \ref{condition1}-\ref{condition2} are satisfied by all processors in $\mathcal{N}$ at global clock $k$, then $\forall i,m\in\mathcal{N}$ and $m\neq i$,
\begin{equation}
u_m^{[i]}(k_i)-\mathcal{U}_m(\pi^{[i]}(k_i))\in O(\sqrt{\epsilon}).
\end{equation}
\end{lemma}

\begin{proof}
Strict concavity of $J(\cdot)$ follows that $\vert J^{[i]}(k_i)-J^{[j]}(k_j)\vert\ge \sigma\Vert z^{[i]}(k_i)-z^{[j]}(k_j)\Vert^2$ for some $\sigma>0$. 
Take $\pi$ as row vector, so the transposition symbol on $\pi$ is omitted, and for simplicity reason we omit $k_i$ and $k_m$ in the following proof.
	\begin{align*}
	&\quad u_m^{[i]}-\mathcal{U}_m(\pi^{[i]})
	=u_m^{[i]}-\min_{p_m\in\mathcal{P}_m}\left\{f_m(p_m)+\pi^{[i]}p_m\right\}\\
	&=u_m^{[i]}-\min_{p_m\in\mathcal{P}_m}\left\{f_m(p_m)+\pi^{[m]}p_m+(\pi^{[i]}-\pi^{[m]})p_m\right\}\\
	&\overset{}{\le} u_m^{[i]}-\mathcal{U}_m(\pi^{[m]})+\max_{p_m\in\mathcal{P}_m}\left\{(\pi^{[m]}-\pi^{[i]})p_m\right\}\\
	&\overset{(1)}{\le}
	u_m^{[m]}+\sqrt{\epsilon/ \sigma}-\mathcal{U}_m(\pi^{[m]})+\max_{p_m\in\mathcal{P}_m}\left\{(\pi^{[m]}-\pi^{[i]})p_m\right\}
	\\
	&\overset{(2)}{\le} \epsilon+\sqrt{\epsilon/\sigma}+\max_{p_m\in\mathcal{P}_m}\left\{(\pi^{[m]}-\pi^{[i]})p_m\right\}\\
	&\overset{}{\le} \epsilon+\sqrt{\epsilon/\sigma}+\max_{p_m\in\mathcal{P}_m}\Vert\pi^{[m]}-\pi^{[i]}\Vert\Vert p_m\Vert\\
	&\overset{(3)}{\le} \epsilon+\sqrt{\epsilon/\sigma}(1+\sqrt{T}\max_{j}\left\{\overline{p}_j\right\})\in O(\sqrt{\epsilon})
	\end{align*}
where (1) comes from $\vert u_m^{[i]}-u_m^{[m]}\vert\le \sqrt{\epsilon/\sigma}$, (2) comes from Condition \ref{condition2} and (3) comes from $\Vert \pi^{[m]}-\pi^{[i]}\Vert\le \sqrt{\epsilon/\sigma}$ and $\Vert p_m \Vert\le \sqrt{T}\bar{p}_m\le \sqrt{T}\max_{j}\left\{\bar{p}_j\right\}$.
\end{proof}

Now, proof of theorem \ref{err_bound} and \ref{pro_feasi} is given:

\begin{proof}
Skip the $k_i$ for conciseness. According to Lemma \ref{C_lemma1}, there exists a positive constant $c$ such that $u_m^{[i]}-\mathcal{U}_m(\pi^{[i]})\le c\sqrt{\epsilon}$ for any $m\neq i$. Then define
\begin{subequations}
	\begin{align}
	&\delta z:=[0,\ldots,0,c\sqrt{\epsilon},\ldots,c\sqrt{\epsilon}]^{\mathsf{T}}\\
	&\bar{z}:=z^{[i]}-\delta z
	\end{align}
\end{subequations}
Note that $\Vert\bar{z}-z^{[i]}\Vert=c\sqrt{n\epsilon}$, so $\bar{z}$ lies in the $c\sqrt{n\epsilon}$ neighborhood of point $z^{[i]}$. According to Lemma \ref{C_lemma1}, we know that for any $m\in\mathcal{N}$ there are $u_m^{[i]}-c\sqrt{\epsilon}\le \mathcal{U}_m(\pi^{[i]})$ for any $m$, so $\bar{z}\in\cap_{i=1}^n\mathcal{Z}_i=\mathcal{S}$. This implies that the intersection of $\mathcal{S}$ and the  $c\sqrt{n\epsilon}$ neighborhood of $z^{[i]}$ which is denoted by $N_{c\sqrt{n\epsilon}}(z^{[i]})$, is not empty.

Suppose that $J^*$ is the maximizer of $J(\cdot)$ in $\mathcal{S}$. Since $\bar{z}\in N_{c\sqrt{n\epsilon}}(z^{[i]})\cap\mathcal{S}$, $J^*\ge J(\bar{z})$. Since $z^{[i]}$ is the maximizer of $J(\cdot)$ in $\mathcal{H}_{tmp}^{[i]}$ and $\mathcal{S}\subset\mathcal{H}_{tmp}^{[i]}$, we have $J(z^{[i]})\ge J^*$. So we get the following relationship:
\begin{equation}
J(z^{[i]})\ge J^*\ge J(\bar{z})\label{inequality}
\end{equation}
Look at the first and last item in (\ref{inequality}):
\begin{subequations}
	\begin{align}
	&\quad J(z^{[i]})-J(\bar{z})\\
	&=e^{\mathsf{T}}(z^{[i]}-\bar{z})+\rho(\Vert\bar{z}\Vert-\Vert z^{[i]}\Vert)\\
	&\le e^{\mathsf{T}}(z^{[i]}-\bar{z})+\rho\Vert\bar{z}-z^{[i]}\Vert\\
	&=nc\sqrt{\epsilon}+\rho c\sqrt{n\epsilon}\in O(\sqrt{\epsilon}).
	\end{align}
\end{subequations}
According to squeeze theorem, we get $0\le J(z^{[i]})-J^*\le $ and $J(z^{[i]})-J^*\in O(\sqrt{\epsilon})$,which completes the proof of theorem \ref{err_bound} and \ref{pro_feasi}.
\end{proof}

\end{document}